\documentclass[a4paper,11pt]{article}
\usepackage{amssymb,latexsym,amsmath,amscd,amsthm}
\usepackage{comment}
\usepackage{epsfig}

\newtheorem{thm}{thm}[section]

\newtheorem{corollary}[thm]{Corollary}

\theoremstyle{defi}
\newtheorem{defi}[thm]{defi}

\newcommand{\N}{\mathcal{N}}
\newcommand{\R}{\mathcal{R}}
\newcommand{\V}{\mathcal{V}}

\usepackage{color}

\definecolor{pink}{rgb}{1,0,1}

\definecolor{garnet}{RGB}{210,15,30}

\title{Rank conditions on phylogenetic networks}
\author{Marta Casanellas, Jes\'us Fern\'andez-S\'anchez\footnote{Departament de Matem\`atiques, 
Universitat Polit{\`e}cnica de Catalunya, Av.~Diagonal 647, 08028 Barcelona. E-mail addresses: \texttt{marta.casanellas@upc.edu} and \texttt{jesus.fernandez.sanchez@upc.edu}}}

\begin{document}

\maketitle

\abstract{
Less rigid than phylogenetic trees, phylogenetic networks allow the description of a wider range of evolutionary events.  %
In this note, we explain how to extend the rank invariants from phylogenetic trees to phylogenetic networks evolving under the general Markov model and the equivariant models.}

\section{Introduction and preliminaries}

In order to model the evolution of a set of DNA sequences (each representing a species), one usually considers a phylogenetic tree (whose leaves are in correspondence with the living species and interior nodes correspond to ancestral species) and a Markov process governing the substitution of nucleotides on it. In phylogenetics, \emph{invariants} is the name given  to the polynomials that vanish on every distribution that arises as a Markov process on the phylogenetic tree. The main idea behind finding invariants is that they might help to distinguish phylogenetic trees and phylogenetic networks and they have been successfully used in phylogenetic reconstruction (see \cite{chifmankubatko2014,casfer2016}), in solving the identifiability of certain models \cite{allman2009} and in model selection \cite{kedzierska2012}.

Nevertheless, trees might be too restrictive to represent the evolutionary history as they cannot take into account processes such as hybridization or horizontal gene transfer. In order to incorporate them, one can use phylogenetic networks.  Invariants for phylogenetic networks have been found for the JC69 substitution model \cite{grosslong} (for networks with a single reticulation vertex) and for the 2-state symmetric model on networks with four leaves \cite{mitchellPhD,mitchell2018}.

We restrict to tree-child binary networks \cite[\S 10]{steelbook}. That is, throughout the paper a \emph{phylogenetic network} $\N$  is a rooted acyclic directed graph (with no edges in parallel) satisfying:
\begin{enumerate}
 \item the root $r$ has out-degree two,
 \item every vertex with out-degree zero has in-degree one and is called a \emph{leaf},
 \item all other vertices have either in-degree one and out-degree two (which are called \emph{tree vertices}) or in-degree two and out-degree one (which are called \emph{reticulation vertices})
 \item the child of a reticulation vertex is a tree vertex.
\end{enumerate}

Following \cite{grosslong} and \cite{nakhleh2011}, we introduce Markov processes on phylogenetic networks.
We denote by $\V$ the set of vertices of the network and will assume that there is a discrete random variable assigned to each vertex taking values in $\Sigma:=\{A,C,G,T\}$. We assign a distribution $\pi=(\pi_A,\pi_C,\pi_G,\pi_T)$ to the root $r$ and to each edge $e$, a 4$\times$4-transition matrix $M^e$. We write $\theta$ for the whole set of these parameters.

Let $\N$ be an $n$-leaf phylogenetic network and associate a $4 \times 4$ transition matrix from a nucleotide substitution model to each {directed} edge of $\N$. Suppose $\N$ has $m$ reticulation vertices $\R=\{w_1, \ldots, w_m\}$.
Each $w_i$ has indegree two, and we denote by $e^0_i$ and $e^1_i$ the two edges directed into $w_i$. Figure \ref{fig_network} shows an example of a phylogenetic network with 4 leaves and only one reticulation vertex $w_1$ (painted white).

Each binary vector $\sigma \in \{0, 1\}^m$ encodes the possible choices for the reticulation edges, where a 0 or a 1 in the $i$-th coordinate indicates that the edge $e^0_i$ or $e^1_i$ was deleted, respectively. Any $\sigma$ results in a $n$-leaf tree $T_{\sigma}$ rooted at $r$ with a collection of transition matrices corresponding to the particular edges in that tree.
We call $\theta_{\sigma}$ the restriction of the parameters $\theta$ of the network to $T_{\sigma}$.
%
%The r depending on the parameters $\theta_{\sigma}=\{\theta^e\}_e$ and corresponding probability distribution $P_{T_{\sigma} ,\theta_{\sigma}} \in {\mathcal M}_T$.
% Namely, the joint probability of the $n$-tuple $(i_1,\ldots,i_n)$ of characters in $\Sigma$ at the leaves of $T_{\sigma}$ is given by
% \begin{eqnarray*}
% (P_{T_\sigma,\theta_{\sigma}})_{i_1,\ldots,i_n}=\sum_{i_w\in \Sigma; w\in \V} \pi_{i_r}
% \left ( \prod_{w \in \V \setminus \R} M^{e_w}_{i_{p(w)},i_w} \right )
% \prod_{j=1}^ m\left (M^{e^0_j}_{i_{p(j,0)},i_j} \right )^{1-\sigma_j} \; \left (M^{e^1_j}_{i_{p(j,1)},i_j} \right )^{\sigma_j}
% % \prod_{w\in V^R} \left (\theta^{e^0_w}_{i_{p(w,0)},i_w} \right )^{1-\sigma_w} \; \left (\theta^{e^1_w}_{i_{p(w,1)},i_w} \right )^{\sigma_j}
% \end{eqnarray*}
% where the entries of the matrices $M^e$ depend on the parameters $\theta^e$, and where for each vertex $w\in \V$, we denote by $p(w)$ the parent vertex of $w$ if $w$ is not a reticulation vertex, and by $p(j,0)$ (resp. $p(j,1)$) the parent node corresponding to the edge $e_j^0$ (resp. $e_j^1$) if $w_j$ is a reticulation vertex.

\begin{figure}[h]
\begin{center} \label{fig_network}
 \includegraphics[scale=1]{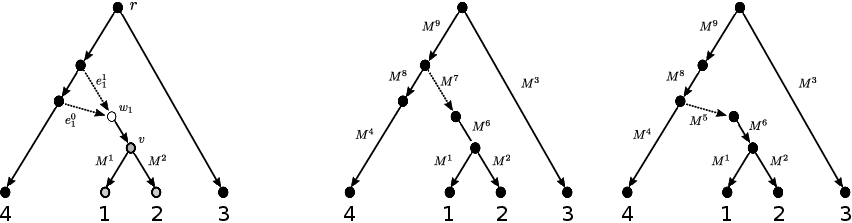}
\caption{On the left, a 4-leaf phylogenetic network $\N$ with one reticulation vertex $w_1$ painted white. The clade corresponding to leaves $A=\{1,2\}$ has been coloured with gray. On the right, the two trees obtained when removing the edges $e^0_1$ and $e^1_1$ incident on $w_1$.}
\end{center}
\end{figure}

For $1 \leq  i \leq m$,
%independently delete $e^0_i$ with denote by probability $\delta_i \in [0, 1]$, otherwise, delete $e^1_i$ . Intuitively,
denote by  $\delta_i$ the parameter corresponding to the probability that a particular site was inherited along edge $e^1_i$.
We can then define a distribution on the set $\Sigma^n$ (corresponding to characters at the leaves of the network) as follows
\begin{eqnarray*}
P_{\N,\theta}=\sum_{\sigma\in \{0,1\}^m} \left (\prod_{i=1}^m \delta_i^{1-\sigma_i}(1-\delta_i)^{\sigma_i}\right ) P_{T_{\sigma} ,\theta_{\sigma}}.
\end{eqnarray*}
% \red{que es $M_T$?} on the leaf states.
% \red{Crec que els $T_{\sigma} $, $P_{T_{\sigma} ,\theta_{\sigma}}$ s'han d'introduir abans, parlar una mica del procés en un arbre i introduir els models --com a minim el GMM}

%We do not restrict the number of reticulation vertices, but we need that all trees $T_{\sigma}$ share a cherry.

% We can adapt this construction by taking the root distribution and the transition matrices in some well-defined evolutionary model $\MM$, which results in considering only distribution vectors and Markov matrices with some particular structure, usually reflected in some identities between coordinates and entries. In this way, we can consider phylogenetic networks on the Jukes-Cantor model, the Kimura 2P (or 3P) model, and so on. When no particular structure is considered, then we have the general Markov model (GMM for short). In the last section of this note we will come back to this situation. \blue{potser moure aquesta paragraf a la darrera seccio}

\begin{defi}
Let $A|B$ be a bipartition of the set of leaves of $\N$.
Given a distribution vector $p$ on  $4^n$ states, the flattening of $p$ relative to the bipartition $A|B$ is the $4^{|A|}\times 4^{|B|}$ matrix $flatt_{A|B}(p)$ whose $(\textbf{i},\textbf{j})$-entry is given by $p(\textbf{k})$ where $\textbf{k}=(\textbf{i},\textbf{j})$ has entries matching those of $\textbf{i}$ and $\textbf{j}$ in the convenient order.
%Introduir $flatt_{A|B}(p)$ of a distribution $p$ on $4^{|A|}\times 4^{|B|}$ states; es pot fer com a AllmanRhodesTaylor: entry $\mathbf{i},\mathbf{j}$ of $flatt_{A|B}(p)$ is
%In the particular case where $|A|=2$, we will simplify the notation and write $flatt_{A}(p)$ for the corresponding flattening.
\end{defi}

%It is well known that one can reroot the tree at any vertex. Reroot $T$ at the vertex of $e$ that is adjacent to $1$ and $2$.
Let $T$ be a tree and let $A|B$ be a bipartition of the leaves of $T$ induced by removing an edge $e$ of $T$. Let $w$ be the vertex of $e$ adjacent to $A$. If $p$ is a distribution on $T$ given by a distribution $\pi$ at $w$ and transition matrices at the edges of $T$ oriented out from $w$, then  $flatt_{A|B}(p)$ can be written as (\cite{AllmanRhodeschapter4},\cite{eriksson2005})
% \begin{equation}\label{eq_flat}
% (M_1\otimes M_2)^t flatt_{12|3\dots n}(q)M_3\otimes \ldots \otimes M_n,
% \end{equation}
\begin{equation}\label{eq_flat}
flatt_{A|B}(p)=(M_A)^t D_{\pi}  M_B,
\end{equation}
where $D_{\pi}$  is the $4\times 4$ diagonal matrix with the entries of $\pi$ at the diagonal,  $M_A$ is the $4\times 4^{|A|}$ matrix whose entry $(x,\textbf{i})$  is the probability in the subtree $T_A$ of observing $\textbf{i}$ at the leaves $A$ given that the node $w$ is at state $x$   (and similarly for $M_B$).
In the next sections we extend the well known edge invariants to phylogenetic networks. On a separate work we will study the consequences that this may have in distinguishing phylogenetic networks and phylogenetic trees.

\section{Invariants for the general Markov model}
%\subsection{}
% \red{Posar figura  de network amb 4  (ex. fig 4 in grosslong) marcant el clade $T_A$ corresponent a la cherry. Citar tambe aquesta figura quan introduim les networks}.

Assume that there is a clade $T_A$ in $\N$ that does not contain any reticulation vertex (this is illustrated in the network of Figure \ref{fig_network}, where the clade $T_A$  corresponds to leaves 1 and 2). Thus $T_A$ is a subtree of $\N$ shared by all $T_{\sigma}$ and the transition matrices at the edges of $T_A$ are also shared by all $T_{\sigma}$. We call $B$ the leaves in $\N$ that are not in~$A$.

%   \red{Caldria que tinguessin la mateixa orientacio? Les matrius del subarbre han de ser exactament les mateixes, això imposa alguna condició en l'orientació dels $T_{\sigma}$?}

\begin{thm}
If $p=P_{\N,\theta}$ is a distribution on a phylogenetic network $\N$ evolving under the GMM and $T_A$ is a tree-clade in $\N$, then $flatt_{A|B}(p)$ has rank $\leq 4$.
\end{thm}

\begin{proof}
 Let $v$ be the root of $T_A$.
 To keep the proof simple we asume that $v$ is different from $r$.
 By rerooting each $T_{\sigma}$ at $v$, the edges of $\N$ that are not in $T_A$ might change their orientation, but the corresponding transition matrices can also be changed so that the joint distribution does not change. If $\mu_{\sigma}$ is the new set of parameters for $T_{\sigma}$, which is composed of the distribution $\pi^{\sigma}$ at the vertex $v$ and the new transition matrices, then $P_{T_{\sigma} ,\theta_{\sigma}} =P_{T_{\sigma} ,\mu_{\sigma}}$.
 %\red{(cal dir que s'han de transposar al fer rerooting ???...)}.
 %Then we have
%\begin{eqnarray*}
%  flatt_{A|B}(p)& =& \sum_{\sigma}(\prod_{i=1}^m \delta_i^{1-\sigma_i}(1-\delta_i)^{\sigma_i}) flatt_{A|B}(P_{T_{\sigma} ,\theta_{\sigma}} ) \\
%   & =&\sum_{\sigma}(\prod_{i=1}^m \delta_i^{1-\sigma_i}(1-\delta_i)^{\sigma_i}) flatt_{A|B}(P_{T_{\sigma} ,\mu_{\sigma}} ).
%  \end{eqnarray*}
Note that after the rerooting process, the new transition matrices associated to the clade $T_A$ are still the same for all $T_{\sigma}$ (even if the distribution $\pi^{\sigma}$ at $v$ might be different for each $T_{\sigma}$). For each $T_{\sigma}$, we write $M_{A}$ for the transition matrix from $v$ to the leaves in $A$  and write $M_B^{\sigma}$ for the transition matrix from $v$ to the leaves in $B$ (as in equation \eqref{eq_flat}). Then, we have
\begin{eqnarray*}
flatt_{A|B}(p) &=& \sum_{\sigma}\left(\prod_{i=1}^m \delta_i^{1-\sigma_i}(1-\delta_i)^{\sigma_i}\right) flatt_{A|B}(P_{T_{\sigma} ,\mu_{\sigma}}) \\
& =&\sum_{\sigma}\left(\prod_{i=1}^m \delta_i^{1-\sigma_i}(1-\delta_i)^{\sigma_i}\right) M_A^tD_{\pi^\sigma}M_B^\sigma= \\ & = & M_A^t \sum_{\sigma}\left(\prod_{i=1}^m \delta_i^{1-\sigma_i}(1-\delta_i)^{\sigma_i}\right) D_{\pi^\sigma}M_B^\sigma,
 \end{eqnarray*}
where the second equality is obtained by using  \eqref{eq_flat} for each $T_{\sigma}$.
Therefore, $flatt_{A|B}(p)$ factorizes  as a product of a $4^{|A|}\times 4$ and a $4\times 4^{|B|}$ matrix, and hence has rank $\leq 4$.
\end{proof}

\begin{corollary}\
If $\N$ is a phylogenetic network with a tree-clade $T_A$ as above and $p$ is a distribution coming from a Markov process on $\N$, then the $5\times 5$ minors of $flatt_{A|B}(p)$ are invariants for $\N$.
\end{corollary}

Note that these invariants are shared by all the phylogenetic networks that have the same clade $T_A$. It is necessary to prove that the $5\times 5$ minors above do not vanish for other networks before using them with the idea of distinguishing networks.

\section{Invariants for equivariant models}

The construction of the first section stands for the general Markov model (GMM), where no particular structure is assumed for the transition matrices or the root distribution.
% By making the evolutionary model more restrictive, one can consider evolutionary submodels of the GMM introduced above.
%
This construction can be adapted by taking the substitution model more restrictive and considering evolutionary submodels of the general Markov model.
%
% This is done by taking the root distribution and the transition matrices in some well-defined evolutionary model $\MM$, which results in considering only distribution vectors and Markov matrices with some particular structure, usually reflected in some identities between coordinates and entries. In this way, we can consider phylogenetic networks on the Jukes-Cantor model, the Kimura 2P (or 3P) model, and so on. When no particular structure is considered, then we have the general Markov model (GMM for short). In the last section of this note we will come back to this situation. \blue{potser moure aquesta paragraf a la darrera seccio}
%
A large class of these submodels are the $G$-equivariant models, where the transition matrices satisfy some symmetries according to a permutation group $G<\mathcal{S}_4$. With precision, equivariant models only consider transition matrices that remain invariant after permuting rows and columns according to the permutations of some given permutation group (see \cite{draisma2008} and \cite{casfer2010} for details). Among the $G$-equivariant models one finds the well known Jukes-Cantor model, Kimura 2 and 3 parameters and the strand symmetric model.

The result obtained in the previous section can be extended to $G$-equivariant models by using the tools introduced in \cite{casfer2010}. We explain briefly the idea.
{Let $\N$ be a network with a tree-clade $T_A$. If $p$ is a distribution on $\N$ arising from a $G$-equivariant model, then $p$ actually lies   in $(\mathbb{C}^{4^n})^G$}, the set of points that remain invariant under the action of $G$. If we write $N_i$ for the irreducible representations of $G$, the regular representation of $G$ induces
a decomposition of $W=\mathbb{C}^4$ into isotypic components: $W\cong \bigoplus_{i=1}^k N_i\otimes \mathbb{C}^{m_i}$, for some well-defined multiplicities $m_i\geq 0$, and similiar decompositions for every tensor power $W^{\otimes l}$, $l\geq 1$ (Maschke's theorem). 
%\red{decompositions of $\mathbb{C}^{4^l}$ (for each $l\geq 1$) into isotypic components.}
%Let $\N$ be a network with a tree-clade $T_A$ and let $p$ be a distribution vector in $(\mathbb{C}^{4^n})^G$. Write $N_i$ for the irreducible representations of $G$. The regular representation of $G$ on the vector space $W=\mathbb{C}^4$, induces decompositions of $\mathbb{C}^{4^n}$ into isotypical components
%\begin{eqnarray*}
%\mathbb{C}^{4^l}=\bigoplus_{i=1}^k N_i\otimes \mathbb{C}^{m_i(l)},
%\end{eqnarray*}
%\red{for certain $m_i(l)\in \mathbb{N}$.
If $|\cdot|$ stands for cardinality, we can rewrite $flatt_{A|B}(p)$ in a convenient basis of $(\mathbb{C}^{4^n})^G\cong \mathrm{Hom}_G(W^{\otimes |A|},W^{\otimes |B|})$ consistent with these decompositions, so that the resulting matrix becomes block diagonal: $$\overline{flatt}_{A|B}(p)=(B_1,\ldots,B_k).$$
%By rewriting $flatt_{A|B}(p)$ in a convenient basis of $\mathbb{C}^{4^n}$, consistent with this decomposition, the resulting matrix becomes block diagonal: $\overline{flatt}_{A|B}(p)=(B_1,\ldots,B_k)$.

%
In this setting, we are able to prove the following result:

%    Ordinary theorem and proof
\begin{thm}
If $p$ arises from the $G$-equivariant model on $\N$, then $\mathrm{rank}(B_i)\leq {m_i}$ for each $i=1,\ldots,k$.
\end{thm}

\begin{corollary}\
If $\N$ is a phylogenetic network with a tree-clade $T_A$ as above and $p$ is a distribution coming from a Markov process on $\N$, then the $({m_i}+1) \times ({m_i}+1)$-minors of the block $B_i$ of $flatt_{A|B}(p)$ are invariants for $\N$.
\end{corollary}

% If $\N$ is a network with a tree-clade $T_A$ and  $p$ is a distribution arising from a $G$-equivariant model on $\N$, then there are invariants for $\N$ that arise from rank conditions on $flatt_{A|B}(p)$. %
% Namely, by Maschke theorem, the regular representation of $G$ on $\mathbb{C}^{4^n}$ induces a decomposition $\mathbb{C}^{4^n}=\bigoplus_{i=1}^k V_i$ into isotypical components, each one corresponding to an irreducible represntation of $G$. Then, $rank (B_i)$ is less than or equal to the multiplicity of $\mathbb{C}^4$ relative to the irreducible representation.
% %
% Namely, by rewriting $flatt_{A|B}$ in a basis that respects isotypical components \red{dir-ho de forma adequada!!!}, the corresponding blocks on the resulting block diagonal matrix must satisfy certain rank conditions depending on the group $G$.

The precise technical statement and the proof will be provided in a forthcoming paper. It will be interesting to check whether these invariants arising from rank conditions coincide with some of the invariants found in \cite{grosslong} for the Jukes-Cantor model.

\subsection*{Acknowledgements} Both authors are partially funded by AGAUR Project 2017 SGR-932 and MINECO/
FEDER Projects MTM2015-69135 and MDM-2014-0445.
%\begin{proof}
%
% text of proof
%\end{proof}

%    Figure insertion
%\begin{figure}
%\includegraphics{filename}
%\caption{text of caption}
%\label{}
%\end{figure}
% You can also use a package such as \epsfig or \psfig.
% Please use \epsfig rather than \epsf, which has
% become obsolete.

%    Mathematical displays

% Numbered equation
%\begin{equation}
%\end{equation}

% Unnumbered equation
%\begin{equation*}
%\end{equation*}

% Aligned equations
%\begin{align}
% &  \\
% &
%\end{align}

%    Bibliography.
%\bibliographystyle{alpha}
%\bibliography{biblio}

\end{document}